\documentclass{article}

\usepackage[final]{neurips_2019}

\usepackage[utf8]{inputenc}
\usepackage[T1]{fontenc}
\usepackage{hyperref}
\usepackage{url}
\usepackage{booktabs}
\usepackage{amsfonts}
\usepackage{nicefrac}
\usepackage{microtype}
\usepackage{enumitem}
\usepackage{graphicx}
\usepackage{xcolor}
\usepackage{lipsum}
\usepackage{float}
\usepackage[square,sort,comma,numbers]{natbib}
\usepackage{subcaption}
\usepackage{mathrsfs}
\usepackage{amsmath,amsthm,verbatim,amssymb,amsfonts,amscd, graphicx, float}
\usepackage{thmtools}
\usepackage{thm-restate}
\usepackage{algorithm2e}

\usepackage{hyperref}
\usepackage[noabbrev]{cleveref}
\Crefname{subsection}{subsection}{subsections}

\newtheorem{theorem}{Theorem}[section]
\newtheorem{corollary}{Corollary}[theorem]
\newtheorem{lemma}[theorem]{Lemma}
\newtheorem{proposition}[theorem]{Proposition}

\theoremstyle{definition}
\newtheorem{definition}{Definition}[section]

\newcommand{\Lspace}{\textbf{L}}
\newcommand{\SLspace}{\textbf{SL}}
\newcommand{\RLspace}{\textbf{RL}}
\newcommand{\NL}{\textbf{NL}}

\newcommand{\STCONN}{\textsc{Stconn}}
\newcommand{\USTCONN}{\textsc{Ustconn}}

\newcommand{\Rot}{\text{Rot}}

\newcommand{\ceil}[1]{\left\lceil #1 \right\rceil}
\newcommand{\zigzag}{\mathbin{\raisebox{.2ex}{
      \hspace{-.4em}$\bigcirc$\hspace{-.75em}{\rm z}\hspace{.15em}}}}
\newcommand{\deranproduct}{\mathbin{\raisebox{.2ex}{
      \hspace{-.4em}$\bigcirc$\hspace{-.75em}{\rm s}\hspace{.15em}}}}

\title{
On the Problem of Undirected st-connectivity\\
}

\author{
  Shilun Li \\
  Dept. of Computer Science \\
  Stanford University \\
  \texttt{shilun@stanford.edu} \\
  \And
  Alex Lee \\
  Dept. of Computer Science \\
  Stanford University \\
  \texttt{leealex@stanford.edu}\\
}

\begin{document}
\maketitle
\begin{abstract}

In this paper, we discuss an algorithm for the problem of undirected st-connectivity that is deterministic and log-space, namely that of Reingold within his 2008 paper "Undirected Connectivity in Log-Space" \cite{reingold2008undirected}. We further present a separate proof by Rozenman and Vadhan of $\USTCONN \in \Lspace$ \cite{rozenman2005derandomized} and discuss its similarity with Reingold's proof. Undirected st-connectively is known to be complete for the complexity class $\SLspace$--problems solvable by symmetric, non-deterministic, log-space algorithms. Likewise, by Aleliunas et. al. \cite{aleliunas1979randomized}, it is known that undirected st-connectivity is within the $\RLspace$ complexity class, problems solvable by randomized (probabilistic) Turing machines with one-sided error in logarithmic space and polynomial time. Finally, our paper also shows that undirected st-connectivity is within the $\Lspace$ complexity class, problems solvable by deterministic Turing machines in logarithmic space. Leading from this result, we shall explain why $\SLspace = \Lspace$ and discuss why is it believed that $\RLspace = \Lspace$.
\end{abstract}

\section{Introduction}
In this report, we shall prove that the problem $\USTCONN$, or st-connectivity on undirected graphs, can be solved with a log-space algorithm. We then explore the implications of such a result. 

We can define the problem of st-connectivity by considering a graph $G$ and two vertices $s$ and $t$ in $G$. The st-connectivity problem answers whether or not the two vertices $s$ and $t$ are connected with each other by a path in $G$. Similarly, the $\USTCONN$ problem decides the $\STCONN$ problem, but on a graph $G$ that is specified to be undirected. (In this case, $\USTCONN$ is a special case of $\STCONN$ and all algorithms that solve $\STCONN$ would be able to solve $\USTCONN$.) The problem of connectivity is one of the most fundamental problems within graph theory, and algorithms to solve $\STCONN$ and $\USTCONN$ has been used to construct more complex graph algorithms. Indeed, a solution to the $\STCONN$ or $\USTCONN$ problem within a certain complexity class would similarly imply a solution within the complexity class for a much larger body of computational problems. 

The time complexity of $\USTCONN$ has been well understood and solved: it is clear to see that the minimum time complexity must be linear (as the length of a path from $s$ to $t$ would be linear), and it is also clear that such a time complexity would be achievable by depth-first search (DFS) or breadth-first search (BFS). 

Most recent study of the $\USTCONN$ problem thus, revolves around its space complexity. It is clear that the space complexity of $\USTCONN$ must be at least $\log$ space, which is the space required to store any $O(n)$ sized variables within the problem. In 1970, Savistch provided a $\log^2{n}$ space complexity solution to the $\STCONN$ (and $\USTCONN$). A randomized algorithm of log-space complexity was also developed in 1979 by Aleliunas, Karp, Lipton, Lovasz, and Rackoff \cite{aleliunas1979randomized}. Following the result, work has been done on derandomizing the randomized algorithm in hopes of creating a deterministic algorithm with decreasing space complexity. In 1992, Nisan, Szemeredi and Wigderson presented an algorithm of $\log^{1.5}$ space \cite{nisan1992log1.5}. In 1999, Armoni, et, al showed that $\USTCONN$ can be solved by an algorithm in $\log^{\frac{4}{3}}$ space \cite{Armoni2000AnOS}. In 2005, Trifonov developed an algorithm of $\log{n} \log\log{n}$ space for $\USTCONN$ \cite{Trifonov2005logloglog}. Finally, in 2008, Omer Reingold presented a deterministic algorithm that solves $\USTCONN$ in log-space complexity \cite{reingold2008undirected}. 

We shall begin our paper by presenting the result of $\USTCONN \in L$ through Omer Reingold's method in his paper "Undirected Connectivity in Log-Space". \cite{reingold2008undirected}. More specifically, we shall begin by explaining expander graphs and some transformations used to convert any graph to an expander graph. We shall then show that these transformations can be performed in log space and that connectivity can be computed from these expander graphs in log space too. This would prove that $\USTCONN   \in \Lspace$.

We would then present a separate proof from Rozenman and Vadhan of $\USTCONN \in \Lspace$ and discuss its similarity to Reingold's proof.

Finally, we shall explore the implications of this result on the relations between the complexity classes of $\Lspace$, $\SLspace$, and $\RLspace$. $\Lspace$ refers to problems solvable by a deterministic log space Turing machine, $\SLspace$ refers to problems solvable by symmetric log space Turing machines, and $\RLspace$ refers to problems solvable by probabilistic log space Turing machines with one sided error.  These three complexity classes are closely tied to $\USTCONN$ and we shall show that $\SLspace=\Lspace$ and discuss why it is believed that $\RLspace = \Lspace$.

\section{Preliminaries}
In this section, we will introduce some properties of graphs using adjacency matrix representation, along with procedures such as graph powering.
\subsection{Graph Adjacency Matrix}
For any graph $G$, common representations include adjacency list, adjacency matrix, and incidence matrix. There exist log-space algorithms which transforms between the common representations, so the problem of $\USTCONN$ does not rely on the input graph representations. In this paper, we will use the adjacency list representation.

\begin{definition}
The \textbf{adjacency matrix} $A$ of a graph $G=(V,E)$ is the $|V|\times |V|$ matrix such that the entry $(u,v)$ of $A$ written $A_{u,v}$ is equal to the number of number of edges from vertex $u$ to vertex $v$ in $G$.
\end{definition}

We allow $G$ to contain self loops and parallel edges.
\begin{definition}
For a graph $G=(V,E)$ with adjacency matrix $A$. $G$ is \textbf{undirected} if $A$ is symmetric, where we have $A=A^T$. An undirected graph $G$ is \textbf{D-regular} if there are exactly $D$ edges incident to every vertex, equivalently, $\sum_{v\in V}A_{u,v}=D$ for all $u\in V$.
\end{definition}

For any undirected D-regular graph $G$ with $N$ vertices, let us label each outgoing edge of every vertex of $G$ by a number from 1 to D in a fixed way. Then we define the rotation map of $G$ as follows:
\begin{definition}
For an undirected D-regular graph $G$ with $N$ vertices, let the \textbf{rotation map} $\Rot_G$ be a permutation of $[N]\times[D]$ defined by $\Rot_G(v,i)=(w,j)$ if edge $i$ from $v$ leads to $w$ and is the same edge as edge $j$ of $w$.
\end{definition}
The rotation map defines how the vertices and edges of $G$ are labeled. The rotation map will play a crucial role in transforming any undirected graph $G$ into a regular graph. The adjacency matrix can be expressed by the rotation map in the following way:
\begin{equation}
    A_{u,v}=|\{(i,j)\in[D]^2:\Rot_G(u,i)=(v,j)\}|
\end{equation}
To solve $\USTCONN$, we would like the graph to be highly connected but at the same time sparse so that the diameter is small. We call such highly connected sparse graphs expanders. We will define expanders using properties of its adjacency matrix.
\begin{proposition}
For an undirected D-regular graph $G$ with $N\times N$ adjacency matrix $A$, $A$ is diagonalizable with eigenvalues $\lambda_1\geq...\geq \lambda_N$. Furthermore, $\lambda_1=D$ and $\lambda_N\geq -D$.
\end{proposition}
\begin{proof}
The first part of the statement follows from spectral theorem for symmetric matrices. For the second part of the proposition, given any eigenvector $\lambda$ of $A$ with eigenvector $v=(v_1,...,v_n)$, consider the index $k$ such that $|v_k|$ achieves the maximum among all $v_1,...,v_i$. Now since $Av=\lambda v$, the $k^{th}$ component satisfies
$$|(Av)_k|=|\sum_{i=1}^n A_{k, i}v_i|\leq \sum_{i=1}^n |v_k|A_{k,i}=D|v_k|.$$
But we also have $|(Av)_k|=|\lambda||v_k|$. So $|\lambda|\leq D$ for any eigenvalue $\lambda$ of $A$.
Now note that for the vector $v=(1,...,1)^T$, we have $Av=Dv$. So $v$ is a eigenvector of $A$ with eigenvalue $D$. Thus $\lambda_1=D$ and $\lambda_N\geq -D$ as desired.
\end{proof}
Now, let us define the \textbf{normalized adjacency matrix} $M$ of an undirected D-regular graph $G$ as the adjacency matrix divided by $D$.
\begin{definition}
For any graph $G$, let $\lambda(G)$ be the second largest eigenvalue of the normalized adjacency matrix. $G$ is an \textbf{expander} graph if $\lambda(G)\leq \frac{1}{2}$.
$G$ is an \textbf{$(N,D,\lambda)$-graph} if it is undirected D-regular with $N$ vertices and $\lambda(G)<\lambda$.
\end{definition}
The second largest eigenvalue of $G$ captures its expansion properties. It is shown by Alon\cite{alon1986eigenvalues} that second-eigenvalue expansion is equivalent to the standard vertex expansion. In particular, we have the following 
\begin{proposition}
\label{expand length}
Fix any fixed $\lambda<1$, for any $(N,D,\lambda)$-graph $G_N$. For any two vertices $s,t\in G_G$, there exists a path of length $O(\log N)$.
\end{proposition}
\begin{proof}
By the result of Alon\cite{alon1986eigenvalues}, for any $\lambda<1$, there exist $\epsilon>0$ such that for any $(N,D,\lambda)$-graph $G_N$ and any set $S$ of vertices of $G_N$ such that $|S|\geq \frac{N}{2}$, we have $|\partial S|\geq \epsilon |S|$ where $\partial S=\{(u,v)\in E(G_N): u\in S, v\in V(G_N)\setminus S\}$. Now for any two vertices $s,t\in G_N$, for some $l=O(\log N)$ with constant only depending on $\epsilon$, since the edge expansion factor is at least $\epsilon$, both $s$ and $t$ can have more than $\frac{N}{2}$ vertices of at most distance $l$. Then there exist a vertex $v$ within distance $l$ from both $s$ and $t$. Thus there exist a path of length $2l=O(\log N)$ from $s$ to $t$. 
\end{proof}
One may notice that the vertex expansion property of an undirected $(N,D,\lambda)$-graphs with $\lambda<1$ implies it is connected. We can also directly see this as for any undirected D-regular graph with multiple connected components, each component will contribute an orthogonal eigenvector of the normalized adjacency matrix with eigenvalue 1 via the indicator of that component. Thus the graph will have the second largest eigenvalue of normalized adjacency matrix being 1 if it has more than one connected component.  With the result above, we can now solve the undirected st-connectivity problem for constant-degree expanders using log-space.
\begin{lemma}
\label{algorithm regular}
For any fixed $\lambda<1$, there exist a space $O(\log D\log N)$ algorithm $\mathcal{A}_\lambda$ such that on an input of an undirected D-regular graph $G$ with $N$ vertices and two vertices $s,t\in G$:
\begin{itemize}
    \item Outputs "connected" only if $s$ and $t$ are connected in $G$.
    \item If $s$ and $t$ are in the same connected component which is a $(N',D,\lambda)$-graph, then the algorithm outputs "connected".
\end{itemize}
\end{lemma}
\begin{proof}
Consider the algorithm of simply enumerating all paths of length $l$ from $s$, where we take $l=O(\log N)$ given by Proposition\ref{expand length}, with the constant only depending on $\lambda$. Such enumeration can be done via the ordering of edges of each vertex. \\
The algorithm outputs "connected" if there is a path of at most length $l$ to $t$. The algorithm runs in $O(l\log D)=O(\log D\log N)$ space as each edge of a vertex requires $\log D$ space and each path of length $l$ can be stored in $O(\log N)$ space. The algorithm satisfies the requirements stated above due to Proposition\ref{expand length}.
\end{proof}
By the explicit construction given by Alon and Roichman \cite{alon1994random} using Cayley graph of the group $\mathbb{F}_2^m$ which is $m$ dimensional vector space of the field $\mathbb{F}_2$, we have existence of expander graphs with desired parameters. 
\begin{proposition}
There exist an undirected $D_0$-regular $((D_0)^{16},D_0,\frac{1}{2})$-graph for some $D_0$.
\label{exist D0 graph}
\end{proposition}
The value $1-\lambda(G)$ is called the spectral gap of a graph. We have shown that for a disconnected graph, the spectral gap is 0. Due to result by Alon \cite{alon2000bipartite}, the converse holds for non-bipartite graphs:
\begin{proposition}[\textbf{Alon}]
\label{lambda upper bound}
For every D-regular connected non-bipartite graph $G$ with $N$ vertices, the spectral gap is at least $\frac{1}{DN^2}$. Equivalently, $\lambda(G)\leq 1-\frac{1}{DN^2}$.
\end{proposition}
\begin{proof}
This is Theorem 1.1 in \cite{alon2000bipartite}.
\end{proof}

\subsection{Graph Powering and Zig-zag Products}
We will now introduce operations of graphs to change its degree and spectral gap, i.e. its expansion properties. We will first introduce graph powering, which reduces its second eigenvalue and increases its spectral gap, but also increases its degree. We will then define the zig-zag product of two graphs, which was first introduced by Reingold, Vadhan and
Wigderson \cite{reingold2000entropy}. This operation reduces the degree of a graph without significantly varying the spectral gap.\\

Recall that the labeling of edges of an undirected D-regular graph is given by the rotation map. Equivalently, the graph is defined by the rotation map. So let us define graph powering via rotation maps.
\begin{definition}
For a undirected $D$-regular graph $G$ with $N$ vertices given by the rotation map $\Rot_G$, the \textbf{$t^{th}$ power} of of $G$ is the $D^t$-regular graph $G^t$ defined by the rotation map $$\Rot_{G^t}(v_0,(a_1,...,a_t))=(v_t,(b_t,...,b_1))$$
for any $v_0,v_t\in [N]$ and $a_1,...,a_t\in [D]$ where $b_1,...,b_t$ are computed by $(v_i,b_i)=\Rot_G(v_{i-1},a_i)$.
\end{definition}
One can view the vector $(a_1,...,a_t)\in [D^t]$ as a path from $v_0$ to $v_t$ where each $a_i$ is the action of taking edge $a_i$ of the current vertex during traversal of the path. The vector $(b_t,...,b_1)\in [D^t]$ is simply the same path backwards, starting from $v_t$ and ending in $v_0$. This definition coincides with the usual definition of graph powering where two vertices is adjacent in the $t^{th}$ power if there is a path of length $t$ in the original graph.

\begin{proposition}
The normalized adjacency matrix of $G^t$ is given by $M^t$ where $M$ is the normalized adjacency matrix of $G$. Consequently, if $G$ is a $(N,D,\lambda)$-graph, then $G^t$ is a $(N,D^t,\lambda^t)$-graph.
\label{power spectral gap}
\end{proposition}
\begin{proof}
From the above discussion, for any two vertices $v_0, v_t\in G^t$, the number of edges between $v_0$ and $v_t$ in $G^t$ is equal to the number of length $t$ paths from $v_0$ to $v_t$, where paths are defined using edges instead of vertices. The number of paths from $v_0$ to $v_t$ is in turn equal to $(A^t)_{v_0,v_t}$, where $A$ is the adjacency matrix of $G$. So $A^t$ is the adjacency matrix of $G^t$. Thus the normalized adjacency matrix of $G^t$ is given by $D^{-t}A^t=M^t$ where $M$ is the normalized adjacency matrix of $G$.\\
If $G$ is a $(N,D,\lambda)$-graph with normalized adjacency matrix $M$. Since the normalized adjacency matrix of $G^t$ is given by $M^t$, we have $\lambda(G^t)=\lambda(G)^t$. Thus $G^t$ is a $(N,D^t,\lambda^t)$-graph.
\end{proof}

For a $(N,D,\lambda)$-graph, powering increases the spectral gap exponentially, but the degree of the graph also increases exponentially. On the other hand, the zig-zag product reduces the degree of the graph but remains the spectral gap nearly unchanged. 

\begin{definition}
Let $G$ be a D-regular graph on $[N]$ with rotation map $\Rot_G$, $H$ be a d-regular graph on $[D]$ with rotation map $\Rot_H$. Then their \textbf{zig-zag product} $G\zigzag H$ is a $d^2$-regular graph on $[N]\times [D]$ with rotation map $\Rot_{G\zigzag H}$ defined by:
$$\Rot_{G\zigzag H}((v,a),(i,j))=((w,b),(j',i'))$$
where $w,b,j',i'$ satisfies: there exist $a',b'\in [N]$ such that
\begin{itemize}
    \item $\Rot_H(a,i)=(a',i')$
    \item $\Rot_G(v,a')=(w,b')$
    \item $\Rot_H(b',j)=(b, j')$
\end{itemize}
\end{definition}
Reingold et al. showed that this product is well defined and $\lambda(G\zigzag H)$ is bounded as a function of $\lambda(G)$ and $\lambda(H)$.\cite{reingold2000entropy} This product replaces each vertex of $G$ via a copy of $H$. The edges of the product correspond to length 3 paths with two edges in $H$ and the middle edge in $G$. See Figure \ref{zigzag} below:

\begin{figure}[H]
\begin{center}
\includegraphics[width=10cm]{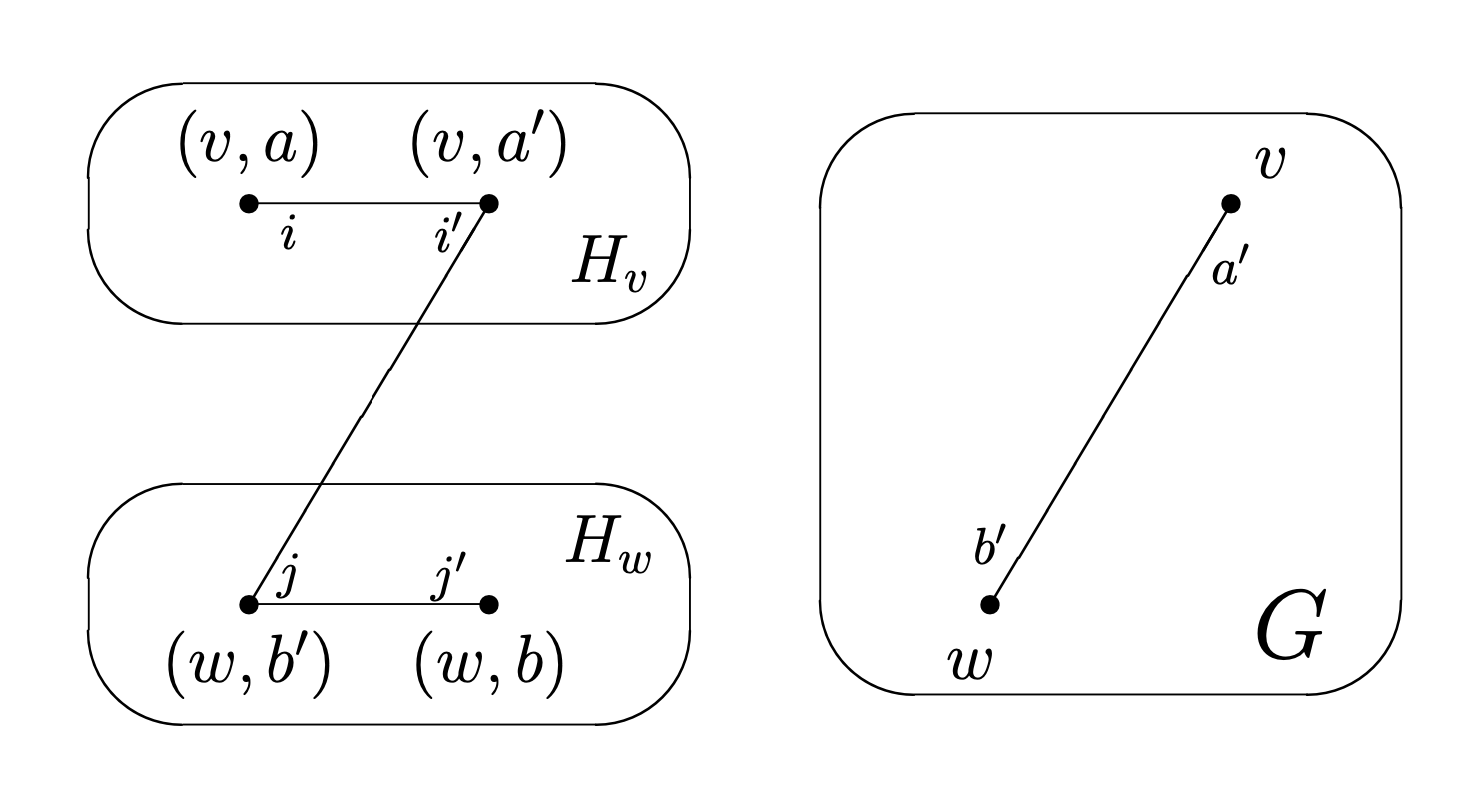}
\end{center}
\vspace*{8pt}
\caption{On the left hand side is the edge $((v,a),(i,j))$ in $G\zigzag H$. $H_v$ is the copy of $H$ for vertex $v$ and $H_w$ is the copy of $H$ for vertex $w$. The edge $((v,a),(i,j))$ correspond a length 3 path composed by $(a, i)$ in $H_v$, $(v,a')$ in $G$ and $(b', j)$ in $H_w$. The rotation map of the edge correspond to the same path traversing backwards. On the right hand side is the projection of the path on $G$, which correspond to the middle edge of the length 3 path.}
\label{zigzag}
\end{figure}
To reduce the degree of the graph $G$ while remaining the spectral gap, we want $d$, the degree of $H$ to be much smaller than $D=deg(G)$ so that $d^2=deg(G\zigzag H)$ is smaller than $D$. We also want $1-\lambda(G\zigzag H)>k(1-\lambda(G))$ for some constant $k\in(0,1)$ independent of $G$. These estimates on the spectral gap of the zig-zag product are given by Reingold et al. in \cite{reingold2000entropy}.

\begin{theorem}[\textbf{Reingold}]
\label{zigzag lower bound}
If $G$ is an $(N,D,\lambda)$-graph and $H$ is a $(D, d, \alpha)$-graph, then $G\zigzag H$ is an $(ND, d^2, f(\lambda,\alpha))$ graph where
$$f(\lambda, \alpha)=\frac{1}{2}(1-\alpha^2)\lambda+\frac{1}{2}\sqrt{(1-\alpha^2)^2\lambda^2+4\alpha^2}$$
\end{theorem}
\begin{proof}
This is Theorem 4.3 in \cite{reingold2000entropy}.
\end{proof}
\begin{corollary}
\label{spectral zig zag lower bound}
If $G$ is an $(N,D,\lambda)$-graph and $H$ is a $(D, d, \alpha)$-graph, then
$$1-\lambda(G\zigzag H)\geq \frac{1}{2}(1-\alpha^2)(1-\lambda).$$
\end{corollary}
\begin{proof}
Since $\lambda\leq 1$, we have
$$\frac{1}{2}\sqrt{(1-\alpha^2)^2\lambda^2+4\alpha^2}\leq \frac{1}{2}\sqrt{(1-\alpha^2)^2+4\alpha^2}=1-\frac{1}{2}(1-\alpha^2).$$
This Corollary is then a direct consequence of Theorem \ref{zigzag lower bound}.
\end{proof}

\section{Expander Transforms of Graphs}
In this section, we will introduce the \textbf{Main Transform} given by Reingold\cite{reingold2008undirected} which uses log-space to transform each connected component of a graph into an expander. This is the main part of the log-space algorithm for \USTCONN. 

\begin{definition}
\label{main transform}
For a $D^{16}$-regular graph $G$ on $[N]$ and $D$-regular graph $H$ on $[D^{16}]$. Then let the transformation $\mathcal{T}$ outputs the rotation map of $G_l$ where $G_l$ is defined recursively by:
$$G_i=(G_{i-1}\zigzag H)^8,\quad i=1,...,l$$
where $G_0=G$ and $l=2\ceil{\log DN^2}$. We will denote $\mathcal{T}_i(G,H)=G_i$ and $\mathcal{T}(G,H)=G_l$.
\end{definition}

From the properties of zig-zag product and graph powering, the graph $G_i$ is a $D^{16}$-regular graph over $[N]\times ([D^{16}])^i$. If $D$ is constant, then $l=O(\log N)$ and $G_l$ has $poly(N)$ vertices. We will first show that this transformation can transform $G$ into an expander.

\begin{lemma}
For a $D^{16}$-regular connected and non-bipartite graph $G$ over $[N]$, and $(D^{16},D,\lambda)$ graph $H$ with $\lambda\leq \frac{1}{2}$, we have $\lambda(\mathcal{T}(G,H))\leq \frac{1}{2}$.
\label{transform lambda upper bound}
\end{lemma}

\begin{proof}
Since $G$ is connected and non-bipartite, by Proposition \ref{zigzag lower bound},
$$\lambda(G_0)\leq 1-\frac{1}{DN^2}.$$
By Corollary \ref{spectral zig zag lower bound}, as $\lambda(H)\leq \frac{1}{2}$, we have
$$\lambda(G_{i-1}\zigzag H)\leq 1-\frac{3}{8}(1-\lambda(G_{i-1}))<1-\frac{1}{3}(1-\lambda(G_{i-1})),\quad i=1,...,l.$$
So by Proposition \ref{power spectral gap} we have
$$\lambda(G_i)=\lambda((G_{i-1}\zigzag H)^8)<(1-\frac{1}{3}(1-\lambda(G_{i-1})))^8.$$
When $\lambda(G_{i-1})\leq \frac{1}{2}$, we have $\lambda(G_i)< (\frac{5}{6})^8<\frac{1}{2}$. If we have $\lambda(G_i)\leq \frac{1}{2}$ for some $i=0,...,l$, then be induction, we would have $\lambda(\mathcal{T}(G,H))=\lambda(G_l)\leq\frac{1}{2}$ as desired. So let us suppose otherwise, $\lambda(G_{i})> \frac{1}{2}$ for all i. Then it is easy to show $\lambda(G_i)=(1-\frac{1}{3}(1-\lambda(G_{i-1})))^8\leq \lambda(G_{i-1})^2$. So 
$$\lambda(G_l)\leq (1-\frac{1}{DN^2})^{2^l}.$$
Since $(1-\frac{1}{x})^x<e^{-1}$ for all $x\geq 1$, we have
$$\lambda(G_l)\leq (1-\frac{1}{DN^2})^{2^l}\leq e^{-2}<\frac{1}{2}.$$
So $\lambda(\mathcal{T}(G,H))\leq \frac{1}{2}$.
\end{proof}

While the analysis in the previous lemma assumes $G$ is connected and non-bipartite, we will extend this analysis of $\mathcal{T}$ to any undirected graph $G$. Note that zig-zag product $G\zigzag H$ and graph powering operates separately on each connected component, $\mathcal{T}(G, H)$ operates on each connected component of $G$ separately. Let define the restriction of a graph:

\begin{definition}
For any graph $G$ and subset $S$ of its vertices $V$. Let $G|_S$ be the the subgraph of $G$ induced by $S$, which has vertices $S$ and edges arising from edges in $G$ which has both endpoints in $S$. 
\end{definition}

Note that $S$ is a connected component of $G$ if $G|_S$ is connected and $S$ is disconnected to vertices of $V\setminus S$. Now let us show that restriction to a connected component of $G$ commutes with taking the transformation $\mathcal{T}$. A crucial observation is that both $\mathcal{T}_i(G|_S,H)$ and $\mathcal{T}_i(G,H)$ are $D^{16}$-regular, with the same vertices. In addition, $\mathcal{T}_i(G|_S,H)$ is a subgraph of $\mathcal{T}_i(G,H)$. So we should have $\mathcal{T}(G|_S, H)=\mathcal{T}(G, H)|_{S\times ([D^{16}])^l}$. An formal proof using induction is given by Reingold which makes use of this observation.

\begin{lemma}
For a $D^{16}$-regular graph $G$ on $[N]$ and $D$-regular graph $H$ on $[D^{16}]$. If $S$ is a connected component of $G$, then
$$\mathcal{T}(G|_S, H)=\mathcal{T}(G, H)|_{S\times ([D^{16}])^l}.$$
\label{restrict}
\end{lemma}

\begin{proof}
 See Lemma 3.3 of \cite{reingold2008undirected}.
\end{proof}

Finally, we will show that $\mathcal{T}$ can be computed in log-space when $D$ is constant. This is essentially due to the fact that during each step of the inductive calculation of $\mathcal{T}_i(G,H)$, only constant addition amount of memory is needed. We will shot the space complexity of $\mathcal{T}$ in the follow lemma:

\begin{lemma}
For any constant $D$. Consider a $D^{16}$-regular graph $G$ over $[N]$ and $D$-regular graph $H$ over $[D^{16}]$, then $\mathcal{T}(G,H)$ can be computed in $\log(N)$ space. Equivalently, there exist an $O(\log N)$ space algorithm $\mathcal{A_T}$ on input $(G,H,(v,a))$ outputs $\Rot_{\mathcal{T}(G,H)}(v,a)$ where $v\in[N]\times ([D^{16}])^l$ and $a\in [D^{16}]$.
\end{lemma}

\begin{proof}
The algorithm $\mathcal{A_T}$ will first allocate variables $v\in[N]$ and $a_0,...,a_l\in[D^{16}]$. We will denote each $a_i=k_{i,1}...k_{i,16}$ where $k\in [D]$ correspond to edge labels of $H$. Now on input $(G,H,(v_{in},a_{in}))$, the algorithm will first copy $v_{in}=(\hat{v},\hat{a}_0,...,\hat{a}_{l-1})\in [N]\times[D^{16}]^{l}$ into the allocated variables $v,a_0,...,a_{l-1}$ and $a_{in}=\hat{a}_l\in [D^{16}]$ into $a_l$. These variables will store the the output of $\Rot_{G_i}$ on $\hat{v},\hat{a}_0,...,\hat{a}_i$ where $G_i=\mathcal{T}_i(G,H)$. We will recursively update the variables $v,a_0,...,a_l$ such that after the $i^{th}$ iteration, the variables $v,a_0,...,a_i$ will store the result of $\Rot_{G_i}((\hat{v},\hat{a}_0,...,\hat{a}_{i-1}),\hat{a}_i)$. For the base case, when $i=0$, $G_0=G$, so we can search in the input tape for the edge $(\hat{v}, \hat{a}_0)$ and write down $\Rot_{G_0}(\hat{v}, \hat{a}_0)$ in $v,a_0$. Now for $i=1,...,l$, we evaluate $\Rot_{G_i}((\hat{v},\hat{a}_0,...,\hat{a}_{i-1}),\hat{a}_i)$ via the following procedure:\\

\textbf{For j = 1 to 16:}
\begin{itemize}
    \item Set $a_{i-1},k_{i,j}=\Rot_H(a_{i-1},k_{i,j})$
    \item If $j$ is odd, recursively compute and set $v,a_0,...,a_{i-1}=\Rot_{G_{i-1}}((v,a_0,...,a_{i-2}),a_{i-1})$.
    \item If $j=16$, reverse the order of the labels in $a_i$: set $k_{i,1},...,k_{i,16}=k_{i,16},...,k_{i,1}$
\end{itemize}
The first two operations correspond to finding a a path of length eight on $G_{i-1}\zigzag H$, which is a step on $(G_{i-1}\zigzag H)^8$. The third bullet reverses the order of labels of $a_i$ to fit the definition of zig-zag and powering. The correctness of the induction follows from the definition of zig-zag product and powering. Thus the correctness of $\mathcal{A_T}$ follows from the inductive definition of $\mathcal{T}$.\\
Now note that within each level of the recursion tree, there are at most $16$ recursive calls, and the recursion tree has depth $l+1=O(\log N)$. So we can maintain the recursive calls with $O(\log N)$ space. Furthermore, the operations of evaluating $\Rot_G$, $\Rot_H$ and reversing labels can be done in $O(\log N)$ space. The space required to store the variables in $O(\log N)$ as $v$ requires $\log N$ space and $a_i$ can be stored in constant space. Thus the total space needed to store $a_i$'s is $O((l+1)*\log 16)=O(\log N)$. Therefore the algorithm $\mathcal{A_T}$ runs in $O(\log N)$ space.
\end{proof}

\section{\texorpdfstring{$\USTCONN\in\Lspace$}{TEXT}}
In this section, we will provide an log-space algorithm for $\USTCONN$ using the by transforming the input graph $G$ into an appropriate expander. 
\begin{theorem}
\label{Reingold proof}
For any undirected graph $G$ over $[N]$, there exists an $O(\log N)$ space algorithm $\mathcal{A}_{con}$ which computes $\USTCONN(G,s,t)$ where $s,t\in [N]$.
\end{theorem}
\begin{proof}
As we can transform between common representations of graphs in log-space, without loss of generality we can assume $G$ is given via the adjacency matrix representation.\\
By Proposition \ref{exist D0 graph}, for some constant $D_0$, there exists a $((D_0)^{16}, D_0, \frac{1}{2})$-graph $H$. Let us hard-code the rotation map of $H$ to the memory of $\mathcal{A}_{con}$. This takes only constant memory.\\
Now we would like to transform $G$ into $D^{16}$-regular graph $G_{reg}$ (which is defined by its rotation map) so that we can apply $\mathcal{T}$ on $(G_{reg},H)$. Let $G_{reg}$ be the graph constructed by replacing each vertex of $G$ with with a cycle of length N, and there is an edge between $(v,w)$ and $(w,v)$ in $G_{reg}$ if there is an edge between $v$ and $w$ in $G$. Self loops are added so that the degree of each vertex is $D_0^{16}$. The rotation map $\Rot_{G_{reg}}:([N])^2\times [D_0^{16}]\rightarrow ([N])^2\times [D_0^{16}]$ is given by:
$$\Rot_{G_{reg}}((v,w),i)=\begin{cases}
((v,(w+1)\mod N), 2),\quad i=1\\
((v,(w-1)\mod N), 1),\quad i=2\\
((w,v),3),\quad i=3\text{ and there is an edge between v and w in G}\\
((v,w),3),\quad i=3\text{ and there is an edge between v and w in G}\\
((v,w),i),\quad i=4,...,16\\
\end{cases}.$$
The first two cases are the edges of the cycle of length $N$ for vertex $v$. The last case are the self loops so that $G_{reg}$ is $D_0^{16}$ regular. Also note that every vertex of $G_{reg}$ has self loops, so $G_{reg}$ is non-bipartite. It is easy to see that $v$ and $w$ are in the same connected component of $G$ if and only if $v\times [N]$ and $w\times [N]$ are in the same connected component of $G_{reg}$. This in turn is equivalent to $(v,1)$ and $(w,1)$ are connected in $G_{reg}$.\\
Now let $G_{exp}=\mathcal{T}(G_{reg},H)$, where $l=O(\log N)$ defined in Definition \ref{main transform}. Let $S$ be the connected component of $s$ in $G$. Then $S\times [N]$ is a connected component of $G_{reg}$ where $G_{reg}|_{S\times [N]}$ is non-bipartite $D_0^{16}$-regular. So by Lemma \ref{restrict}, $S\times [N]\times ([D_0^{16}])^l$ is a connected component of $G_{exp}$ and 
$$\mathcal{T}(G_{reg}|_{S\times [N]}, H)=G_{exp}|_{S\times [N]\times ([D_0^{16}])^l}.$$
Thus by Lemma \ref{transform lambda upper bound}, we get
$$\lambda(G_{exp}|_{S\times [N]\times ([D_0^{16}])^l})\leq \frac{1}{2}.$$
Now let us run the $O(\log N)$ space algorithm $\mathcal{A}_{\lambda}$ with $\lambda = \frac{1}{2}$ on $G_{exp}$ and $(s,1^{l+1})$ and $(t,1^{l+1})$ given by Proposition \ref{algorithm regular}. The algorithm $\mathcal{A}_{con}$ will output "connected" if $\mathcal{A}_\lambda$ outputs connected, else it will output "disconnected".\\
The correctness of $\mathcal{A_T}$ follows from the discussion above, as $s$ and $t$ are connected in $G$ if and only if $s\times [N]\times ([D_0^{16}])^l$ and $t\times [N]\times ([D_0^{16}])^l$ are connected in $G_{exp}$. The algorithm runs in $O(\log N)$ space as computing $\Rot_{G_{reg}}$, the main transform $G_{exp}$ and running $\mathcal{A}_{\lambda}$ can be done using $O(\log N)$ space. So $\mathcal{A}_{con}$ computes $\USTCONN(G,s,t)$ using $O(\log N)$ space.
\end{proof}

\section{An alternative proof of \texorpdfstring{$\USTCONN\in\Lspace$}{TEXT}}
This section will contain an alternative proof of $\USTCONN\in\Lspace$ given by Rozenman and Vadhan\cite{rozenman2005derandomized}. In both proofs, the key idea is that $\USTCONN$ is solvable in log-space on bounded-degree graphs with logarithmic diameter by enumerating over all paths. Bounded-degree Expander graphs (graphs with second eigenvalue less than $\frac{1}{2}$) are instances of such graphs, both proofs transform the graph into an expander with bounded degree to solve $\USTCONN$. \\

In the proof of Reingold, given any undirected graph with $N$ vertices, Reingold first transforms the graph into a regular graph and then used a combination of graph powering and zig-zag product to transform the regular graph into an expander graph with constant degree over $poly(N)$ vertices, while maintaining the connectivity properties of vertices. Graph powering increases the connectivity of the graph, decreases $\lambda(G)$ while increasing the degree and number of vertices polynomially. Zig-zag product decreases the degree while keeping $\lambda(G)$ approximately still. The combination of both decreases $\lambda(G)$ to $\frac{1}{2}$ while maintaining the degree constant.\\

One the other hand, Rozenman and Vadhan's proof\cite{rozenman2005derandomized} shares the same overall process as Reingold, but used derandomized squaring instead of graph powering and zig-zag products to increase the connectivity of the graph. Iterating derandomized squaring yields highly connected graphs with relatively small degree compared graph powering while maintain the same number of vertices. 
\begin{definition}
Let $G$ be an undirected $D$-regular graph over $[N]$, let $H$ be an undirected d-regular graph over $[D]$. The derandomized square graph $G\deranproduct H$ is an undirected $Dd$-regular graph over $[N]$ with rotation map
$$\Rot_{G\deranproduct H}(v,(x,a))=(w,(h,b))$$
where
\begin{itemize}
    \item $(u,y) = \Rot_G(v,x)$
    \item $(z,b) = \Rot_H(y,a)$
    \item $(w,h) = \Rot_G(u,z)$
\end{itemize}
for any $v\in[N], x\in[D], a\in[d]$.
\end{definition}
An edge in $G\deranproduct H$ corresponds to a length 2 path in $G$. The following figure illustrates an edge of the derandomized square $G\deranproduct H$:
\begin{figure}[H]
\begin{center}
\includegraphics[width=10cm]{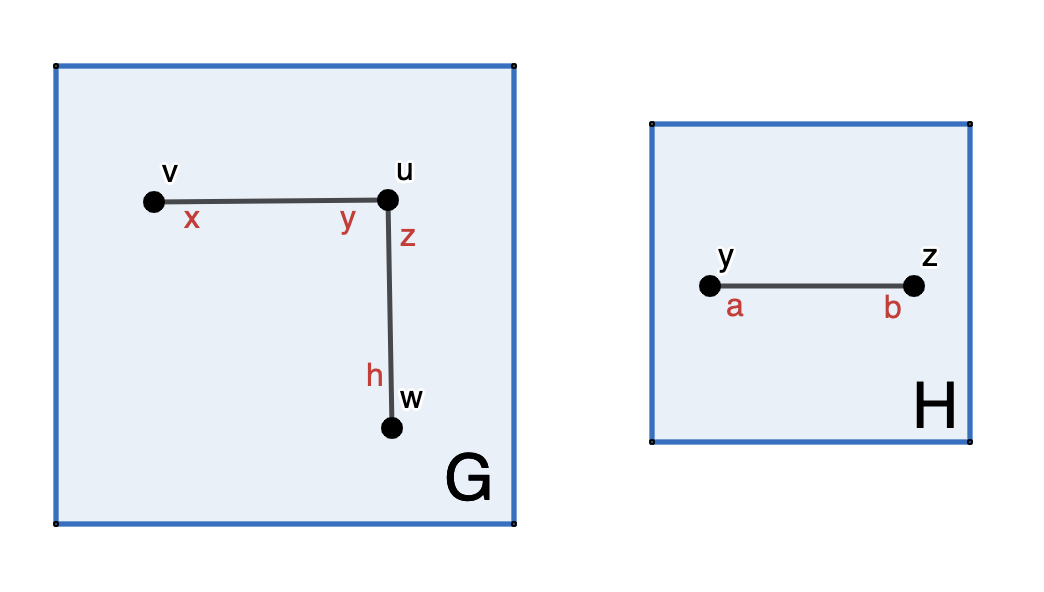}
\end{center}
\vspace*{8pt}
\label{desquare}
\caption{The figure on the left shows the edge $(v,(x,a))\in G\deranproduct H$. Variables in black correspond to vertex labels and red labels corresponds to edge indices.}
\end{figure}
Similar to the zig-zag product, the derandomized product operates separately on the connected components of $G$. The derandomized square $G\deranproduct H$ increases the connectivity of $G$ and increases the spectral gap. The following theorem by Rozenman and Vadhan gives an upper bound on $\lambda(G\deranproduct H)$:
\begin{theorem}[\textbf{Rozenman-Vadhan}]
If $G$ is an undirected $(N,D,\lambda)$-graph and $H$ is an undirected $(D,d,\alpha)$-graph, then $G\deranproduct H$ is an $(N,Dd,f(\lambda,\alpha))$-graph where
$$f(\lambda,\alpha)=1-(1-\lambda^2)(1-\alpha)\leq \lambda^2+\alpha$$
\end{theorem}
\begin{proof}
See Theorem 6.3 in \cite{rozenman2005derandomized}.
\end{proof}
To prove $\USTCONN\in\Lspace$ in a manner similar to Theorem \ref{Reingold proof}, we require a family of undirected constant degree expander graphs to apply derandomized squaring with. In addition, these graphs need to be computed in log-space. This is possible by Reingold\cite{reingold2000entropy} and Gabber\cite{gabber1981explicit}
\begin{lemma}
For some constant $d=16^q$, there exists a family $\{X_m\}_{m\in \mathbb{N}}$ of undirected graphs where $X_m$ is an $(d^m,d,\frac{1}{100})$-graph. Furthermore, $\Rot_{X_m}$ can be computed in space $O(m)$.
\label{exist family}
\end{lemma}
\begin{definition}
Let $X_m$ be the family of constant degree graphs in Lemma \ref{exist family}, let $m_0$ be some fixed constant, define
\begin{align*}
    H_m&= X_m,&\text{when }m\leq m_0\\
    H_m&= (X_{m_0-1+2^{m-m_0}})^{2^{m-m_0}}, &\text{when }m > m_0&
\end{align*}
$\Rot_{G_m}$ can be computed in space $O(m+2^{m-m_0})$.
\end{definition}
With the existence of such family of expanders, Rozenman and Vadhan takes a similar approach as Reingold to prove $\USTCONN\in\Lspace$. For any input graph $G$ over $[N]$, it can be transformed to a 16-regular graph $G_{reg}$, then by powering, it can be turned to a $d$-regular graph $G_0=(G_{reg})^{q}$, where $d=16^q$ is the constant in Lemma \ref{exist family}. Then he recursively defined $G_{m+1}=G_{m}\deranproduct H_{m}$. Taking $m_0$ to be $\ceil{100\log N}$ and $m_1=m_0+\log\log N+10$, Rozenman and Vadhan showed that $\lambda(G_{m_0}|_S)<\frac{3}{4}$ and $\lambda(G_{m_1}|_{S'})\leq \frac{1}{2N^3}$ for each connected component $S$ of $G_{m_0}$ and $S'$ of $G_{m_1}$. In addition, he showed that $G_{m_1}$ has degree $poly(N)$ and can be constructed in $O(\log N)$ space, i.e. $\Rot_{G_{m_1}}$ can be computed using $O(\log N)$ space. As all transformations in this procedure operates separately on each connected component, we can solve $\USTCONN$ of $G$ by running the algorithm of Lemma \ref{algorithm regular} on $G_{m_1}$.



\section{\texorpdfstring{$\SLspace=\Lspace$}{TEXT}}
The section will contain a proof of $\SLspace=\Lspace$. Let us first define the space $\SLspace$.

A Turing machine can be defined by the 7-tuple $(K, \Sigma, \Sigma_0, k, \Delta, s, F)$. Specifically, $K$ is a finite set of states, $\Sigma$ is the finite tape alphabet, $\Sigma_0 \subseteq \Sigma$ is the input alphabet, $k > 0$ is the number of tapes, $s \in K$ is the initial state, $F \subseteq K$ is the set of final states, and $\Delta$ is the finite set of transitions.

Using this definition, we define transitions for the Turing machine with form $(p, (ab, D, cd), q)$ where $a, b, c, d \in \Sigma$ and $D \in {1, -1}$. A transition of the form $(p, ab, 1, cd, q)$ means that if a Turing machine in state $p$, scans $a$ and symbol $b$ is contained in the square to the right of the scanned square, the Turing machine moves the tape head one square to the right, rewrite the squares with symbol $a$ and symbol $b$ with symbol $c$ and symbol $d$, respectively, and changes to state $q$. On the other hand, a transition of the form $(p, ab, -1, cd, q)$ means that if a Turing machine in state $p$ scans symbol $b$ and symbol $a$ is contained in the square to the left of the scanned square, the Turing machine moves the tape head one square to the left, rewrite the square with symbol $b$ and symbol $a$ to symbol $d$ and symbol $c$, respectively, and changes to state $q$. 

For Turing machines with multiple tapes, we define the transition form $\delta = (p, t_1, t_2, \cdots, t_k, q)$, where $k$ is the number of tapes. Each $t_i$ is a 3-tuple $(ab, D, cd)$, which specifies the transition as described above for tape $i$. 

For a non-deterministic Turing machine, there can be multiple transitions from each possible configuration. For each of these possible choices, the non-deterministic Turing machine creates a branch in its configuration path. A non-deterministic Turing machine accepts a configuration if any of the branches within its configuration path ends at an accepting state.

We note that our definition of Turing machine is equivalent to that of a standard Turing machine. Our "peeking" Turing machine can be reduced to the big-headed Turing machine as defined by Hennie \cite{hennie_1979}, which has been shown to be equivalent to a standard Turing machine.

Now, let us define symmetrical Turing machines using the definition by Lewis and Papadimitriou in "Symmetric Space-Bounded Computation" \cite{LEWIS1982161}.

\begin{definition}
For a transition $\delta = (p, t_1, t_2, \cdots, t_k, q)$, with $t_i=(a_ib_i,D_i,c_id_i)$, we define its inverse $\delta^{-1} = (q, t_1^{-1}, t_2^{-1}, \cdots, t_k^{-1}, p)$, where $t_i^{-1} = (c_id_i, -D_i, a_ib_i)$. 
\end{definition}
\begin{definition}
A $\textbf{Symmetrical Turing Machine}$ is a non-deterministic Turing machine whose transition functions $\Delta$ is invariant under taking inverse, namely, for every non-deterministic transition $\sum p_i\delta_i \in \Delta$, we have $\sum p_i\delta_i^{-1} \in \Delta$.
\end{definition}

\begin{definition}
The space $\SLspace$ is the set of all languages which can be determined by a symmetrical log-space Turing machine.
\end{definition}

We note that a symmetric Turing machine has a number of special transitions, from which it is always possible to revert from these transitions (since the symmetric Turing machine includes the inverse of these transitions). 

To prove that $\USTCONN$ is $\SLspace$-complete, we shall use a lemma from the paper Symmetric Space-Bounded Computation by Lewis and Papadimitriou\cite{LEWIS1982161}. We begin by defining relevant terms in the lemma. 

\begin{definition}
If there exists a transition from configuration $C_1$ to $C_2$, we write $C_1 \vdash_M C_2$ or equivalently $C_2 \dashv_M C_1$.
\end{definition}

Let us define the reflexive, transitive closure of $\vdash$, denoted $\vdash^*_M$ and the transitive closure of $\vdash$, denoted  $\vdash^+_M$. For any  $M$, let $\mathcal{A}$ be a subset of all possible configurations on $M$. If for some possible configurations $C_0,...,C_n$ of $M$, we have $C_0 \vdash_M C_1 \vdash_M \cdots \vdash_M C_n$ for some $n \geq 0$ and $C_1, C_2, C_3, \cdots C_n \not \in \mathcal{A}$, we write $C_0 \vdash_M^{*\mathcal{A}} C_n$ (equivalently $C_n \dashv_M^{*\mathcal{A}} C_0$). Note that if $C_0 \in \mathcal{A}$ and $C_0 \vdash_M C_0$, we have $C_0 \vdash_M^{*\mathcal{A}} C_0$. If $A_1 \vdash_M^{*\mathcal{A}} B \vdash_M A_2$ for $A_1, A_2 \in \mathcal{A}$ and $B$ a possible configuration of $M$, we write $A_1 \vdash_M^{+\mathcal{A}} A_2$ (equivalently $A_2 \dashv_M^{+\mathcal{A}} A_1$).

For a Turing machine $M$, define the Turing machine $M*$, which is the same as $M$ except that one can't re-enter its initial state, leave its final state nor write blanks on its tapes. We also define $\overline{M}*$ as the symmetrically closed $M*$, i.e. translations of $\overline{M}*$ is the union of the transition of $M*$ and its inverse.

\begin{lemma}
For a non-deterministic Turing machine $M = (K, \Sigma, \Sigma_0, k, \Delta, s, F)$, let $\mathcal{A}$ be a subset of all possible configurations of $M$. If the following conditions hold:
\begin{enumerate}[label=(\alph*)]
    \item For any $A_1, A_2 \in \mathcal{A}$, if $A_1 \vdash_M^{+\mathcal{A}} A_2$, then $A_2 \vdash_M^{+\mathcal{A}} A_1$.
    \item For any $A$ in the union of $\mathcal{A}$ and possible initial configurations of $M$, any $B \not \in \mathcal{A}$, and any $C_1, C_2, C_3$, if $A \vdash_M^{*\mathcal{A}} C_1 \dashv_M^{*\mathcal{A}} C_2 \dashv_M B \vdash_M C_3$, then $C_2 = C_3$.
    \item For any $A_1$ in the union of $\mathcal{A}$ and possible initial configurations of $M$, any $A_2 \in \mathcal{A}$, and any $B$, if $A_1 \vdash_M^{*\mathcal{A}} B \dashv_M^{*\mathcal{A}} A_2$, then $A_1 = A_2$.
\end{enumerate}
Then, the symmetrical non-deterministic Turing machine $\overline{M}*$ would accept the same language as $M$ in the same space as $M$. 
\label{symmetric lemma}
\end{lemma}
\begin{proof}
This is Lemma 1 in \cite{LEWIS1982161}.
\end{proof}

Using this lemma, we can now prove that $\USTCONN$ is $\SLspace$-complete.
\begin{theorem}
$\USTCONN$ is $\SLspace$-complete.
\label{SL complete}
\end{theorem}
\begin{proof}
We begin by proving that $\USTCONN \in \NL$. In other words, we describe a non-deterministic Turing machine that can solve $\USTCONN$.

Let us define a non-deterministic Turing machine $M$ with 2 tapes. Given an undirected graph $G$ and nodes $s,t\in G$, $M$ begins by writing $s$ and $t$ on its two tapes. Let the tape containing $t$ be the tape containing the destination node and the tape containing $s$ be the tape containing the current node. At the start of each step, we non-deterministically choose a neighbor of the current node, rewrite the neighbor into the tape containing the current node, and check if the node in the tape containing the current node is the same as the destination node. If the current node is the same as the destination node, $M$ accepts, else $M$ continues the process. For our non-deterministic process of choosing a neighbor, we move through the edges from left to right. For each edge, we check if the edge contains the current node. If it does, with probability $\frac{1}{2}$, we update the current node by the other node in the edge. We move from the leftmost edge to the rightmost edge in the input tape to maintain a constant order in choosing neighbors of the current node.

Since both of our tapes only store $1$ node, it is clear that our non-deterministic Turing machine run in log-space. 

Now, we shall show that our non-deterministic Turing machine satisfies the conditions in Lemma \ref{symmetric lemma}. Let $\mathcal{A}$ be the configuration in the Turing machine where the tapes contain $t$ and the current node and the tape head reading the inputs is at the start of an edge (about to choose a neighbor of the current node). 

Now, let us consider condition $(a)$ of the Lemma. For any $A_1, A_2 \in \mathcal{A}$, where $A_1 \vdash_M^{+\mathcal{A}} A_2$, let the current node in configuration $A_1$ be $c_1$ and let the current node in configuration $A_1$ be $c_2$. Since we have $A_1 \vdash_M^{+\mathcal{A}} A_2$, we know that $c_1$ and $c_2$ must be neighbors (since we wouldn't go through another configuration in $\mathcal{A}$ before we arrive at the configuration $A_2$, i.e. we would not be choosing any other node to get to $c_2$). Furthermore, since all the edges in graph $G$ are undirected, it is clear that we can go back from $c_2$ to $c_1$ with the same path. Thus, we have $A_2 \vdash_M^{+\mathcal{A}} A_1$.

Next, let us consider condition $(b)$ of the Lemma. Take any $A$ in the union of $\mathcal{A}$ and possible initial configurations of $M$, $B \not \in \mathcal{A}$, and $C_1, C_2, C_3$, such that $A \vdash_M^{*\mathcal{A}} C_1 \dashv_M^{*\mathcal{A}} C_2 \dashv_M B \vdash_M C_3$. From $A \vdash_M^{*\mathcal{A}} C_1 \dashv_M^{*\mathcal{A}} C_2 \dashv_M B \vdash_M C_3$, it is clear that $B$ is either $A$ or $\not \in \mathcal{A}$. Thus, $C_1, C_2, C_3 \not \in \mathcal{A}$, and are configurations of $M$ when non-deterministically choosing the neighbors of the current node in configuration $A$. Since we always choose neighbors by checking from the leftmost edge to the rightmost, there is only one possible linear process to non-deterministically choose the neighbors of the current node, i.e. for any configuration while choosing neighbors of the current node, $M$ can only be coming from one possible configuration and can only transition to one possible configuration. Thus, it is clear that $C_2 = C_3$.

Finally, let us consider condition $(c)$ of the Lemma. Take any $A_1$ in the union of $\mathcal{A}$ and possible initial configurations of $M$, $A_2 \in \mathcal{A}$, and any $B$, such that $A_1 \vdash_M^{*\mathcal{A}} B \dashv_M^{*\mathcal{A}} A_2$. Let the current node in $A_1$ be $c_1$. It is clear that $B$ is a configuration of $M$ while $M$ is choosing the neighbors for $c_1$. It is also clear that to return to another configuration with a current node that is not $c_1$, $M$ must first return to configuration $A_1$. Thus, $A_1 = A_2$. 

Since $M$ satisfied all three condisions of Lemma \ref{symmetric lemma}, by the lemma, the symmetrical Turing machine $\overline{M}^*$ determines $\USTCONN$ in log-space. So $\USTCONN\in\SLspace$.

Using Savitch's argument in \cite{SAVITCH1970177} and noting that the graph generated from a Symmetric Machine is undirected (since we can revert from any transition), we have that $\USTCONN$ is $\SLspace$-complete.

More specifically, consider any problem in $\SLspace$ which is solved by the symmetric Turing machine $M$. Since $M$ can be computed in log-space, there are polynomially many states for $M$. We can construct a graph $G$, where the nodes are the possible states of $M$, and the edges are transitions between the possible states. Since $M$ is a symmetric Turing machine, we note that the edges are undirected. Thus, solving the problem using $M$ would be equivalent to checking if there's a path that connects the node of a starting state to the node of an accepting state. Thus, we have shown that all problems in $\SLspace$ can be reducible to $\USTCONN$ and $\USTCONN$ is $\SLspace$-complete.
\end{proof}

\begin{theorem}
$\SLspace=\Lspace$
\end{theorem}
\begin{proof}
This is a direct consequence of Theorem \ref{Reingold proof} and Theorem \ref{SL complete}.
\end{proof}

\section{Discussion}
In this section, we will discuss the importance of the paper by Reingold\cite{reingold2008undirected} and some further research based on the paper.

The paper by Reingold has made progress towards discovering the relationship between $\Lspace$ and $\RLspace$. Let us begin by defining $\RLspace$.
\begin{definition}
$\RLspace$ is the space of all languages $L$ such that there exists a randomized Turing machine $T$ which runs in log-space and polynomial time and satisfies
$$P(T\text{ accepts }x)>\frac{2}{3}, \quad x\in L$$
$$P(T\text{ rejects }x)=1, \quad x\not\in L$$
\end{definition}
Notice that we can choose any constant $0<c<1$ replacing $\frac{2}{3}$. We can also increase the probability of accepting $x$ when $x\in L$ to $1-2^{-poly(|x|)}$ by repeating the algorithm $poly(|x|)$ times.



It has been shown by Aleliunas et al, in 1979 that $\STCONN \in \RLspace$ \cite{aleliunas1979randomized} In particular, $\USTCONN$, as a specific case of $\STCONN$ is contained in $\RLspace$. Furthermore, it has been shown that random walks can be generated by a randomized Turing machine of $\RLspace$ in polynomial time and log space. Thus, Reingold's proof that $\USTCONN \in \Lspace$ has brought forth major areas of research into the properties of $\RLspace$.

Building upon this research, Reingold, Trevisan, and Vadhan has shown in 2005\cite{Reingold2005} that a subset of $\STCONN$ ($\STCONN$ for graphs whose random walks are of polynomial mixing time) is $\RLspace$ complete and a deterministic log-space Turing Machine can be used to simulate random walks for biregular graphs.

Some areas of future research into the relationship between $\Lspace$ and $\RLspace$ could be to investigate whether one can describe a deterministic log-space Turing Machine that can simulate random walks for all regular directed graphs. 

\pagebreak
\bibliography{main}
\bibliographystyle{alpha}

\end{document}